\documentclass[accepted]{paper} %

\usepackage[american]{babel}

\usepackage{natbib} %
\usepackage{mathtools} %
\usepackage{booktabs} %
\usepackage{tikz} %

\usepackage{algpseudocode}
\usepackage{algorithm}

\usepackage{amsthm}
\theoremstyle{plain}
\newtheorem{theorem}{Theorem}[section]

\newtheorem{lemma}[theorem]{Lemma}

\theoremstyle{definition}
\newtheorem{definition}[theorem]{Definition}

\theoremstyle{remark}

\DeclareMathOperator*{\argmax}{arg\,max}

\usepackage[textsize=tiny]{todonotes}

\newcommand\blfootnote[1]{%
  \begingroup
  \renewcommand\thefootnote{}
  \footnotetext{#1}%
  \addtocounter{footnote}{-1}%
  \endgroup
}

\usepackage{amsfonts}
\usepackage{cleveref}
\usepackage{multirow}

\usepackage{bbm}

\title{Decentralized Inference via Capability Type Structures \\in Cooperative Multi-Agent Systems}

\author[1]{Charles Jin}
\author[1]{Zhang-Wei Hong}
\author[1,*]{Farid Arthaud} 
\author[1,*]{Idan Orzech}
\author[1]{Martin Rinard}

\affil[1]{%
    CSAIL, MIT, Cambridge, MA 02139
}
\affil[*]{%
    equal contribution
}
  
\begin{document}
\maketitle

\begin{abstract}
This work studies the problem of ad hoc teamwork in teams composed of agents with differing computational capabilities. We consider cooperative multi-player games in which each agent's policy is constrained by a private capability parameter, and agents with higher capabilities are able to simulate the behavior of agents with lower capabilities (but not vice-versa). To address this challenge, we propose an algorithm that maintains a belief over the other agents' capabilities and incorporates this belief into the planning process. Our primary innovation is a novel framework based on capability type structures, which ensures that the belief updates remain consistent and informative without constructing the infinite hierarchy of beliefs.  We also extend our techniques to settings where the agents' observations are subject to noise. We provide examples of games in which deviations in capability between oblivious agents can lead to arbitrarily poor outcomes, and experimentally validate that our capability-aware algorithm avoids the anti-cooperative behavior of the naive approach in these toy settings as well as a more complex cooperative checkers environment.
\end{abstract}

\blfootnote{
    Correspondence to \href{mailto:ccj@csail.mit.edu}{ccj@csail.mit.edu}.
}
\section{Introduction}

Multi-agent systems (MAS) research has seen an increasing interest in developing autonomous agents that can collaborate with others in order to achieve common goals. In many real-world scenarios, agents may have different computational capabilities. For example, in a fleet of autonomous robots, some robots may have better sensors or processing power than others, which can affect their ability to complete the joint objective. As sensors can fail and processing power can fluctuate, it is important to develop techniques that allow for online inference and adaptation to differing capabilities.

Ad hoc teamwork~\citep{mirsky2022survey} has emerged as framework to study the problem of cooperation in MAS, where agents form teams on the fly to solve a given task without prior coordination. Ad hoc teamwork can be beneficial when communication between agents is limited or when agents need to collaborate with unfamiliar or unknown agents. 
However, existing approaches do not account for the differences in computational capabilities between agents, and also generally support only a single adaptive agent.

We propose a capability-aware ad hoc teamwork approach, where agents must adapt to each others' computational abilities. Our setting assumes a hierarchy of capabilities, where agents with stronger capabilities can simulate the policies of agents with weaker capabilities. Our approach builds on prior work on the public agent method \citep{nayyar2013decentralized} for decentralized control as well as Harsanyi's seminal work on type structures \citep{harsanyi1967games} for Bayesian games. In addition, our work extends prior approaches to ad hoc teamwork by taking into account the capabilities of agents. We demonstrate the effectiveness of our approach through experiments in several toy settings and a cooperative checkers task. Our experiments show that our capability-aware algorithms outperform traditional approaches that do not consider computational differences between agents.

\section{Related Work}

The objective of ad hoc teamwork \citep{mirsky2022survey} is designing agents that can adapt to new teammates with limited coordination. However most works focus on the setting of a single ad hoc teamwork agent, with the remainder of agents being non-adaptive. A notable exception is \cite{albrecht2016belief}, in which the teammates may adapt using separate strategies, avoiding the problem of higher order beliefs. To the best of our knowledge, this work is also the first to consider ad hoc teamwork between agents of different computational capabilities.

Our work is also related to teammate/opponent modeling in cooperative multi-agent planning~\citep{torreno2017cooperative}. Cooperative multi-agent planning aims to coordinate a team of agents to maximize a joint utility function. Classic works employ heuristic search algorithms on each agent to maximize the joint utility. For example, \citep{desaraju2011decentralized} coordinate collision-free paths for agents in a team using the rapidly-exploring random tree (RRT) algorithm. In addition to heuristic search, prior works \citep{lowe2017multi,foerster2018counterfactual} in multi-agent reinforcement learning \citep{sutton2018reinforcement} jointly learn policies for each agent, including multi-agent belief-systems \citep{ouyang2016dynamic,foerster2017learning,foerster2019bayesian, vasal2022framework}. In contrast, we focus on the problem of online inference and adaptation, with particular attention to the complexities that arise when the inference procedure is also constrained by capabilities.

Finally, many works consider computationally bounded agents, particular in the domain of human-robot interaction.
Humans are well-known to exhibit sub-optimal behavior, which poses a challenge to works that model humans as perfectly rational. A common approach, called \textit{bounded rationality}, instead assumes agents to be (approximately) optimal subject to certain constraints \citep{russell1991right, simon1997models}. For instance, several works consider agents which have limited capacity to reason recursively about other agents \citep{9658965, wen2019modelling}, and prove convergence assuming the recursive depths in the population follows a Poisson distribution. \citep{nikolaidis2016formalizing} models human teammates as having bounded memory in the sense that they reason over only the most recent $k$ interactions. However, all these works treat the amount of ``boundedness'' as a known quantity; in contrast, our work studies the problem of adapting to an unknown level of boundedness.

\section{Background and Setting}

We consider a fully cooperative stochastic game setting with $N \ge 2$ players. To model players with heterogeneous capabilities, we assume capabilities are drawn from a totally ordered set of known \textbf{capability types} $\mathcal{C}$. For instance, consider the class of policies that combines an exhaustive look-ahead search down to a maximum depth of $k$ with a value function approximation at the leaves. A set of capability types for this class of policies could specify a set of upper bounds on the depth of the search, with stronger capabilities correspond to deeper searches.  We adopt the convention that types lower ($<$) in the hierarchy are weaker in terms of capability.

At the beginning of the game, each player $i$ is randomly assigned a capability type $c_i$ from $\mathcal{C}$. As the type of a player restricts their ability to either play or even consider certain strategies, in general, optimal play under different assignments of capabilities may also require different (joint) policies. Our objective is for players to learn each others' private capability types and converge on an informed joint policy without explicit communication. For clarity of presentation, we will assume that the base game (with public types) is fully observable with a single agent acting in each timestep.

Formally, the game proceeds as a discrete-time partially observable decentralized Markov decision process (Dec-POMDP) \citep{littman1994markov}, given by a tuple $\langle \mathcal{S}, \mathcal{N}, \mathcal{C}, \mathcal{A}, P, R, \rho_0, \gamma \rangle$, where 
$\mathcal{S}$ is the set of all states of the base game, 
$\mathcal{N} = \{1, 2, \ldots, N\}$ is the set of players,
$\mathcal{C}$ is the set of capability types,
$\mathcal{A}$ is the set of available actions,
$P:\mathcal{S}\times\mathcal{A}\to\Delta(\mathcal{S}\times\mathcal{N})$ gives a distribution over state transitions,
$R:\mathcal{S}\times\mathcal{A} \to \Delta([0, 1])$ gives a distribution over (bounded) rewards,
$\rho_0 = \Delta(\mathcal{S}\times\mathcal{N}\times \mathcal{C}^N)$ is a distribution over the initial state, and
$\gamma \in (0, 1)$ is the discount factor.
Each episode of the game begins with an initial state $s(0)$, player $i(0)$, and assignment of capabilities $C = (c_1, \ldots, c_N)$ drawn from $\rho_0$. Each player privately observes their assigned capability type. At timestep $t$, all players observe the state $s(t)$, then player $i(t)$ selects an action $a(t) \in \mathcal{A}$; all players receive a reward $r(t) \sim R(s(t), a(t))$ and transition to the next timestep $t{+}1$ with state $s(t{+}1), i(t{+}1) \sim P(s(t), a(t))$. The objective is to maximize the expected joint return $\mathbb{E}\big[\sum^{\infty}_{t=0}\gamma^t r(t)$\big].

\subsection{Planning in Dec-POMDPs}

In a single-agent POMDP setting where the environment is partially observed, the agent maintains a belief $b$ over the true state of the environment, which is initialized to some prior distribution and subsequently updated via Bayes rule as the game proceeds. Such a belief serves as a sufficient statistic for the agent's history of actions and observations \citep{kaelbling1998planning}. Each action then induces a posterior distribution over not only the next state but also the next belief. We can thus form the corresponding \textit{belief-state MDP} by treating the belief as part of the state, which reduces the POMDP to a standard MDP and permits solving the POMDP using standard techniques such as value iteration.

However, in multi-agent settings, such first-order beliefs about the environment are not enough, because other agents may also have partially observed internal states containing private information. In particular, as a player's actions can depend on privately-held \textit{beliefs}, in order to compute an optimal policy, other agents must maintain second order beliefs about each others' first order beliefs, \textit{ad infinitum}. For instance, interactive POMDPs \citep{gmytrasiewicz2005framework} model such beliefs explicitly to a given depth, but in general this technique scales poorly with the size of the belief space (and number of agents).

Another approach is to instead construct a (fictitious) \textit{public agent} \citep{nayyar2013decentralized}, whose beliefs are conditioned on information observable by all agents. Each player's policy is also assumed to be derived from a public policy $\pi$, which depends only on on the private information $s_i^{pri} \in \mathcal{S}_i^{pri}$ and the public beliefs $b^{pub} \in \mathcal{B}^{pub}$. More explicit, the public policy $\pi$ consists of a sequence of per-timestep action-selection functions $\pi = \{\pi(t) : \mathcal{S}^{pri} \times \mathcal{B}^{pub} \to \Delta(\mathcal{A})\}_{t=0}^\infty$. Player $i$'s policy $\pi_i$ is then defined to be $\pi_i(\cdot) := \pi(s_i^{pri}, \cdot)$. As different private information $s_i^{pri}$ yields different policies in general, this formulation allows the public belief about $s_i^{pri}$ to be updated via Bayes rule after each action by player $i$:
\begin{align}
\label{eq:pub_belief_update}
    P(s_i^{pri} | a, b^{pub}) \propto P(a | s_i^{pri}, b^{pub}) \cdot P(s_i^{pri} | b^{pub})
\end{align}
where $P(a | s_i^{pri}, b^{pub})$ is given by the public policy $\pi(s_i^{pri}, b^{pub})$. As all players can compute the public agent's belief update, i.e., $b^{pub}$ is \textit{common knowledge}, this strategy obviates the need to model higher-order beliefs.

Finally, given one policy per player, we can form the joint policy $\Pi = \{\pi_1, \ldots, \pi_N\}$, which yields a distribution over trajectories (and hence rewards) by selecting each player's actions according to their respective policy. We denote the expected reward of a joint policy $\Pi$ as $J(\Pi) := \mathbb{E}\big[ \sum^{\infty}_{t=0}\gamma^t r(t)\big]$. We denote the belief state for player $i$ as $s_i := (s_i^{pri}, b^{pub})$. The value function of the policy for the joint belief state $s := (s_i)_{i \in \mathcal{N}}$ and time $t$ is
\begin{align}
\label{eq:value_def}
V^\Pi(s, t) &:= \gamma^{-t}\mathbb{E} \big[ \sum^{\infty}_{t'=t}\gamma^{t'} r(t') \mid s \big].
\end{align}

In this work, we will also consider a dual perspective that arises from \textit{value function approximation}. Given an arbitrary value function $V$, the corresponding $Q$ function is
\begin{align}
Q(s, a, t) 
&:= 
\label{eq:q_fn}
\mathbb{E}_{r, s'} \big[ r + \gamma V(s', t{+}1) \mid a \big],
\end{align}
where $s'$ is the next (stochastic) belief state after taking action $a$. One can then define a new, greedy policy $\Pi^V$ by
\begin{align}
\Pi^V(s, t) 
&:= 
\label{eq:value_fn_policy}
\argmax_{a \in \mathcal{A}} Q(s,a,t),
\end{align}
with ties broken randomly. For a value function $V$ that is derived from some existing policy $\Pi$ according to \Cref{eq:value_def}, the Policy Improvement Theorem \citep{Bellman:1957} guarantees that this greedy policy $\Pi^V$ achieves better rewards than the original policy $\Pi$.

\subsection{Capability-Feasible Policies}

We next describe how we formally model players with varying computational capabilities using capability types. In our setting, player $i$'s private information $s_i^{pri}$ consists of their capability type $c_i \in \mathcal{C}$. Their policy under the public agent method is given by $\pi_i(b^{pub}) := \pi(c_i, b^{pub})$, where $\pi$ is a known public policy. Since the policy depends only on the type of the player, it suffices to consider the collection of public \textbf{typed policies} $(\pi^c)_{c \in \mathcal{C}}$, where $\pi^c(\cdot) := \pi(c, \cdot)$.

We say that the policy $\pi^{c}$ is \textbf{feasible} for a player of type $c'$ if $c' \ge c$. Players are allowed to evaluate only the policies that are feasible for their capability type. This prevents the player from ``pretending'' to play as the stronger type $c$ by simply evaluating $\pi^{c}$. For instance, consider again a planning algorithm equipped with a capability type that bounds the lookahead depth. The set of feasible policies are those with equal or lesser lookahead depth.

The introduction of capability types introduces several challenges to the direct application of the public agent method. First, in the original formulation, all players should be able to evaluate all the public policies in order to compute the public belief update in \Cref{eq:pub_belief_update}. However, in our case the restriction of players to evaluating only feasible policies prevents agents from performing the full belief update (which would require computing the posterior probability $P(c | a)$ for all $c \in \mathcal{C}$). A more fundamental problem is that the public beliefs $b^{pub}$ must be computable by all players, and in particular, the weakest player. Hence, the public agent method devolves to using only the \textit{weakest} player's inferences. To contend with this challenge, the next section develops a generalization of the public agent method inspired by \textit{Harsanyi type structures}, which allow us to maintain beliefs for different agents while still avoiding the infinite belief hierarchy.

\section{Capability-Aware Ad Hoc Teamwork}

In this section, we describe a general multi-agent ad-hoc teamwork framework that adjusts for the capabilities of other players over each episode. To start, each player independently maintains a belief over the capabilities of all players, which is updated each time the player observes an action. Players then adapt their play to incorporate their belief about the capability of the other players. We identify two main challenges in this setting. First, since players interact only through the base game, they cannot coordinate their belief updates; hence one technical challenge is to overcome the non-stationarity in the environment from any single player's perspective, as other players are also adjusting their behaviors online. Second, players may, in general, need to consider an infinite hierarchy of beliefs: because other players' beliefs affect their future actions, the current player's action should also incorporate \textit{second order} beliefs, i.e., beliefs about other players' beliefs; but other players now also have second order beliefs, so one needs to maintain third-order beliefs, \textit{ad infinitum}.

To address these challenges, we propose a learning framework based around \textit{capability type structures}, which are beliefs that preserve the capability ordering:

\begin{definition}
Let $(\mathcal{C}, \le)$ be a set of capability types. Given a base set $X$, a \textbf{capability type structure}, denoted $\mathcal{B}_X$, is a collection of real-valued functions $B_X^c : X \to \mathbb{R}$ indexed by $\mathcal{C}$ and \textbf{reduction} operators $\downarrow_{c}$ such that if $c' \le c$ then ${B_X^c\downarrow_{c'}} = B^{c'}$. In this context, we refer to each $B_X^c$ as a belief of type $c$, or a \textbf{typed belief} in general.
\end{definition}

Let $B_X := (B^i_X)_{i \in \mathcal{N}}$ denote the collection of players' beliefs about $X$. Given a collection of reduction operators, we can form the beliefs $B_X$ into a capability type structure $\mathcal{B}_X$ if and only if for all players $i, j \in \mathcal{N}$ either (1) $c_i = c_j$ and $B^i_X$ = $B^j_X$, or (2) $c_j \le c_i$ and $B^i_X \downarrow_{c_j}$ = $B^j_X$ (and similarly when $c_i \le c_j$). If this case, we can identify player $i$'s belief $B^i_X$ with the typed belief $B^{c_i}_X \in \mathcal{B}_X$.

Capability type structures are a natural generalization of the public belief in the sense that the typed belief $B^c_X$ is common knowledge for all players whose type $c' \ge c$. This crucial point allows us to avoid regressing to the weakest player, while still allowing players to update their beliefs in a manner that retains consistency, thereby overcoming the non-stationarity of the other agents' beliefs (and hence, behaviors). Capability type structures are also related to Harsanyi type structures in that the capability type structure $\mathcal{B}_X$ generates all higher-order beliefs that any player holds for some property of interest $X$. In particular, if $B_X^c$ is the belief of player $i$ with type $c$, then $i$'s second order belief about ``what a player $j$ of type $c'$ believes about $X$'' is given by $B_X^c \downarrow_{c'}$, which, by construction, is equivalent to $j$'s first-order belief $B_X^{c'}$; it follows that all the higher-order beliefs are also equivalent to $B_X^{c'}$. In contrast, Harsanyi type structures are, in general, intractable to construct explicitly.

\subsection{Inference With Capability Type Structures}
\label{sec:noise_free_inference}

In this section, we introduce our framework for capability-adaptive multi-agent systems leveraging capability type structures. We will assume that players have access to a collection of value functions $(V^c)_{c \in \mathcal{C}}$, where $V^c(s, B^c)$ takes as input\footnote{For clarity, we will generally hide dependences on $t$, particularly of the value functions.} the current (fully observable) game state $s$, as well as a belief $B^c$ of type $c$, and outputs an estimate of the value. We analogously define a collection of corresponding $Q$ functions $(Q^c)_{c \in \mathcal{C}}$ according to \Cref{eq:q_fn}. Players then query the $Q$ functions to execute the corresponding typed policy $\pi^c$ given by the greedy policy in \Cref{eq:value_fn_policy}.

We begin by describing the initial setup. Each player maintains a belief for each player (including themselves), for a total of $N$ beliefs, where beliefs are vectors of real numbers indexed by the base set $\mathcal{C}$. We denote player $i$'s belief about player $j$ as $B^i_j$, and denote player $i$'s aggregate beliefs as $B^i := ( B^i_j )_{j \in \mathcal{N}}$. Each belief $B^i_j$ holds an unnormalized likelihood over player $j$'s capability type, conditioned on player $j$'s type being less than player $i$'s type. To initialize their beliefs, a player of type $c$ sets $B[c'] = 1$ if $c' \le c$ and 0 otherwise, for all $B \in B^i$. We additionally define a belief reduction operator $\downarrow_{c'}$ that sets all entries of the input belief vector corresponding to capabilities $c > c'$ equal to zero. Note that the collection of beliefs $\mathcal{B}_j := (B^i_j)_{i\in\mathcal{N}}$ trivially forms a capability type structure at initialization, where the base set $\mathcal{C}$ over which the beliefs are expressed represents possible values for player $j$'s type.

We next describe how players adjust their policies based on their beliefs. Given a collection of beliefs $B$, we denote by $B_{[j = c]}$ the same collection except with $B_j$ (the belief about player $j$'s type) replaced with $B_j[c'] = \infty$ if $c' \ne c$, and $0$ otherwise (this intervention conditions the beliefs on player $j$'s true type being $c$). Player $i$ selects the action that maximizes the $Q$ function of their type $c$, yielding the following greedy policy:
\begin{align}
\label{eq:belief_policy}
\pi^i(s, B^i) := \argmax_{a \in \mathcal{A}} Q^{c}(s, B^i_{[i=c]}, a)
\end{align}
with ties broken randomly.

Finally, we describe how the players update their beliefs. Each turn consists of two stages. In the first stage, the acting player $i$ plays an action $a$ according to the policy $\pi^i$. In the second stage, all players $j \in \mathcal{N}$ update their beliefs $B^j_i$ about player $i$. Denote the optimal value at capability $c$ as
\begin{align}
    v^*(c) := \max_{a' \in \mathcal{A}} Q^c(s, B^j_{[i=c]} \downarrow_{c}, a'),
\end{align}
and the set of actions to achieve $v^*(c)$ as
\begin{align}
    A^*(c) := \{a' \in \mathcal{A} | Q^c(s, B^j_{[i=c]} \downarrow_{c}, a') = v^*(c)\}.
\end{align}
Then player $j$ of type $c_j$ performs the update
\begin{align}
\label{eq:update_exact}
B^j_i[c] \leftarrow B^j_i[c] * \mathbbm{1}_{a \in A^*(c)} / |A^*(c)|
\end{align}
for all $c \le c_j$. Notice the update factor is simply the likelihood $P(a|c)$ under the policy $\pi^i$.

Denote the collection of all players' beliefs about player $j$'s type at time $t$ as $B_j(t)$. Our key result is the following:
\begin{theorem}[Consistency of belief updates]
\label{thm:consistency}
If $B_j(t)$ is a capability type structure, then so is $B_j(t{+}1)$.
\end{theorem}
We provide a full proof in the Appendix. As a corollary, since the prior beliefs $B_j(0)$ are initialized as a capability type structure, it follows that $B_j(t)$ can also be collected into a type structure for all $t \ge 0$. 

Our next result states that the beliefs are accurate. We first need to clarify what the beliefs are tracking. For a belief $B_j$ of type $c$, we can define the function
\begin{align}
\label{eq:exact_cond_likelihood}
    P^c(c') := \frac{B_j[c']}{\sum_{c'' \le c} B_j[c'']}, \forall c' \in \mathcal{C}
\end{align}
whenever the denominator is non-zero.
The following theorem states that $P^c(c')$ is in fact the true (conditional) likelihood $P(c_j = c' | c_j \le c)$ as computed by the ``omnipotent'' public agent performing the \textit{exact} Bayesian updates with access to all $Q$ functions $(Q^c)_{c \in \mathcal{C}}$.
\begin{theorem}[Correctness of belief updates]
\label{thm:correctness}
Let $j$ be a player, and denote its type as $c_j$. Let $B_j$ be the belief of type $c$ at time $t$, let $H(t)$ be the full action-observation history as of time $t$, and define $P^c$ as in \Cref{eq:exact_cond_likelihood}. Then whenever $P^c$ is defined, we have that $P^c(c') = P(c_j = c' | c_j \le c, H(t))$ for all $c' \le c$. Furthermore, if $P^c$ is undefined, then $P(c_j \le c | H(t)) = 0$.
\end{theorem}
Note that the true likelihood is not, in general, guaranteed to converge to (a delta distribution at) the true capability absent additional structural assumptions on the base game's MDP.
However by \Cref{thm:correctness}, if the true likelihood does converge, then the beliefs in the capability type structure also converge to the correct conditional likelihoods.

\subsection{Capability-Aware Ad Hoc Teamwork With Noise}
\label{sec:noisy_inference}

We next turn to a setting where players are subject to independently sampled noise. In general, the presence of noise significantly complicates belief updates in multiagent systems, as any small initial difference in beliefs due to noise can lead to observed behaviors diverging from expected behaviors, which in turn amplifies the difference in beliefs, leading to increasingly inconsistent beliefs. The main idea is to perform a \textit{tempered} version of the update in \Cref{eq:update_exact}, and show that the resulting ``likelihoods'' remain consistent even under noise. Assuming that small differences in likelihoods lead to small difference in the $Q$ function, we show that the beliefs can still be assembled in an \textit{approximate} capability type structure, yielding approximate versions of \Cref{thm:consistency}. All proofs are contained in the appendix.

\subsubsection{Ad Hoc Teamwork With Capability Types}

We now introduce our generic model of ad hoc teamwork. We assume that the game is ``solved'' when the types are known, that is, there exist a public collection of value functions $(V^C)_{C \in \mathcal{C}^N}$ ranging over all possible assignments of capabilities $C$. As with before, we can use the value function $V^C$ to define the $Q^C$ function and the one-step greedy policy $\Pi^C$. We say $Q^C$ is feasible for a player $i$ of type $c$ if $c' \le c$ for all $c' \in C$ (that is, player $i$ can only accurately unroll the policy when all the players in the joint policy are at most as capable as $i$), and restrict players to policies with feasible $Q^C$ functions.

To cooperate with stronger agents, we assume in this work that player $i$ substitutes in the strongest type available to it, namely, its own type $c$. Such a setting is natural when computing (or representing) the value function requires computational resources unavailable to the player of lower capabilities, and so players use a \textit{best-effort} approximation to players of greater capabilities. As stronger players are able to exactly simulate the approximations used by weaker players, our setting still allows the team to play a joint policy that is ``optimal'' for their assignment of types $C$, with the strongest player(s) playing $\Pi^C$ and all weaker players using their best-effort approximations. For a player a type $c$, we define the \textbf{predecessor set} as $p(c) := \{c' \mid c' \le c\}$, and denote the set of \textbf{feasible assignments} as $n(c) := p(c)^N$.

To incorporate beliefs into this framework, we adopt the QMDP approximation \citep{littman1995learning}. Fix a \textbf{generalized likelihood function} $\phi$ that takes a collection $B = (B^c_j)_{j \in \mathcal{N}}$ of type $c$ beliefs about players' types, and outputs a distribution over available assignments: $\phi_B(\cdot) \in \Delta(n(c))$. We define the $\phi$-value- and $\phi$-$Q$ functions as
\begin{align}
V^c_\phi(s, B) &:= \mathbb{E}_{C \sim \phi^c_B}\big[ V^C(s) \big], \\
Q^c_\phi(s, B, a) &:= \mathbb{E}_{r, s'} \big[ r + \gamma V^c_\phi (s', B) \mid a \big],
\end{align}
respectively. In our case, the $\phi$ value function can be interpreted as an estimate of the value of an action if all types were to be revealed before the next step. For simplicity we consider only the one-step greedy policy, but in general, exponentially more accurate estimates can be obtained by increasing the look-ahead depth \citep{Bellman:1957}. 

Player $i$ then selects the action that maximizes the $\phi$-$Q$ function of their type, yielding the following greedy policy:
\begin{align}
\label{eq:phi_policy}
\pi^i_\phi(s, B^i) := \argmax_{a \in \mathcal{A}} Q^{c_i}_\phi(s, B^i_{[i=c_i]}, a).
\end{align}

\subsubsection{Noise Model}

Our noise model allows (1) each player's observations of the state $s$ to be subject to independent noise, as well as (2) players to use different, \textit{private} versions of the $V^C$ functions, which captures many natural sources of noise: for instance, observations of the state $s$ may be subject to sensor noise, and allowing players to use private $V^C$ allows our results to apply to the more general setting of \textit{decentralized training and decentralized execution}, as players that independently solve for approximate value functions will differ by at most the sum of their approximation errors.

Our only assumption is that the deviations are bounded. 
Let $s_1, s_2$ and $V_1^C, V_2^C$ be the noisy states and private value functions of any two players, respectively. We assume that
\begin{align}
\label{eq:noise}
    |V_1^C(s_1) - V_2^C(s_2)| \le \epsilon
\end{align}
for a known constant $\epsilon \ge 0$. Note that the noise-free setting is the special case when $\epsilon = 0$, in which case our results from the previous section apply.

\subsubsection{Tempered Belief Updates}

The section provides a general framework for the setting where players' observations and value functions are subject to (bounded) noise. The main problem when moving from the noise-free setting is that the hard update in \Cref{eq:update_exact} is not robust to noise, and will in fact lead to incorrect inferences if applied directly. One method is to incorporate the noise explicitly in the Bayesian update, however this is impractical as the exact form again results in an infinite hierarchy of beliefs; furthermore, in order to perform a closed-form update we would require an assumption that the distribution of noise at the level of the value function is known, which is hard to guarantee in practice, particularly for decentralized training.

Instead, we propose a technique based on \textbf{tempered beliefs}, where the belief update is smoothed by a temperature parameter. Rather than maintain an explicit likelihood as in the noise-free case, it will be more convenient to use the beliefs to track a sufficient statistic. All beliefs are initialized to the all-zeros vector. After observing player $i$ take action $a$, player $j$ of type $c_j$ performs the update
\begin{align}
    B^j_i[c] \leftarrow B^j_i[c] + \ell_a(c)
\end{align}
for all $c \le c_j$, where $\ell$ is the \textit{loss}:
\begin{align}
\ell_a(c) := \max_{a' \in \mathcal{A}} Q_\phi^c(s, B^j_{[j=c]}, a') - Q_\phi^c(s, {B^j_{[j=c]}}, a).
\end{align}

Player $j$ computes the generalized likelihood at time $t$ as
\begin{align}
    P^j_i(c) :=& \frac{\exp(-B^j_i[c] / T(t))}{\sum_{c' < c_j} \exp(-B^j_i[c'] / T(t))},\\
    \phi^c_B(C) \propto& \prod_{i \in \mathcal{N}} P^j_i(C[i]),
\end{align}
where $T: t \rightarrow \mathbb{R}^+$ is the temperature, and executes the policy according to \Cref{eq:phi_policy}. Note that as $T \rightarrow 0$, $P^j_i(c)$ converges to 0 if capability $c$ ever experienced positive loss after observing an action by player $i$, which is equivalent to the hard update in \Cref{eq:update_exact} (save for a normalizing factor). Conversely, as $T \rightarrow \infty$, the updates become uninformative, and the likelihood converges to the initial uniform distribution. Hence, we choose the temperature to be as small as possible (to be informative) while providing the benefit of smoothing out the noise.

\subsubsection{Adversarial Noise}

If the noise is adversarial, then the only constraint on the noise is given by the $\epsilon$ bound in \Cref{eq:noise}. In this case, we set the temperature $T(t) = 6tN$, which yields the following generalized likelihood:
\begin{align}
    \phi^c_B(C) \propto \prod_{i \in \mathcal{N}} \exp(-B^j_i[c_i]/6tN)
\end{align}

\begin{theorem}[Consistency of belief updates, adversarial noise]
\label{thm:consistency_approx_adv}
Let $B^j_i$ and $B^k_i$ denote the type $c'$ beliefs held by players $j$ and $k$ about player $i$'s type. Then $|B^j_i[c] - B^k_i[c]| / t \le 3\epsilon$ for all $c \in \mathcal{C}$ and timesteps $t \ge 0$.
\end{theorem}

\subsubsection{Stochastic Noise}
\label{subsec:update_stoch}

For adversarial noise, we suffered a linear factor of $t$ in the generalized likelihood, essentially taking the average loss. Our next result provides a PAC guarantee for stochastic noise with a sublinear factor in the temperature. In particular, we will further assume that the observed error in the value function is sampled iid for each realization.

Let $\delta > 0$ be given. We choose the temperature $T(t) = \sqrt{d}t^{2/3}$, yielding the following generalized likelihood:
\begin{align}
    \phi^c_B(C) \propto \prod_{i \in \mathcal{N}} \exp(-B^j_i[c_i]/\sqrt{d}t^{2/3})
\end{align}
where $d$ is a constant that depends logarithmically on $1/\delta$.

\begin{theorem}[Consistency of belief updates, stochastic noise]
\label{thm:consistency_approx_stoch}
Let $B^j_i$ and $B^k_i$ denote the type $c'$ beliefs held by players $j$ and $k$ about player $i$'s type. Then with probability at least $1 - \delta$, $|B^j_i[c] - B^k_i[c]| / t^{2/3} \le \sqrt{d}\epsilon/2N$ for all $i,j,k \in \mathcal{N}$, $c \in \mathcal{C}$, and timesteps $t \ge 0$.
\end{theorem}

\section{Experimental Results}

We present results for several experiments using players who cooperate using a depth-bounded online tree search. Specifically, a player of capability $d$ seeks to maximize the joint rewards over the next $d$ steps. Each task consists of two players that alternate taking actions. One player is an \textit{expert}, with a deeper search depth, and the other player is a \textit{novice}, with a shallower search depth. The players are otherwise identical (i.e., have the same action space, share the same reward function, use the same search hyperparameters, and have complete knowledge of the environment as well as the actions taken by the other player).

To plan, agents use the Monte-Carlo Tree Search (MCTS) algorithm on the depth-bounded tree. Because the value estimates at the leaves depend on a Monte Carlo estimate, we use the stochastic version of the tempered beliefs as described in \Cref{subsec:update_stoch} for all the experiments. Our appendix contains a more detailed description of how the MCTS policies and belief updates are defined.

\subsection{Toy Environments}

We first present results in two toy environments to demonstrate the effects of capability-awareness in a cooperative setting. In the \textbf{Wall of Fire} task, two players take turns controlling a single avatar. The avatar can move in one of the four cardinal directions. After each move, the players collect a reward based on their position: the red ``fire'' tiles yield a penalty of \mbox{-2}, whereas the yellow ``coin'' tiles yield a reward of \mbox{+100}. Each coin can only be collected once. We run each episode for 20 turns (10 per player). We use a novice of depth 2, and an expert of depth 20. \Cref{fig:wall_of_fire} displays the initial state of the players and the board.  \Cref{tab:wall_of_fire} reports the performance of various team compositions.

\begin{figure}
    \centering
    \includegraphics[width=.5\columnwidth]{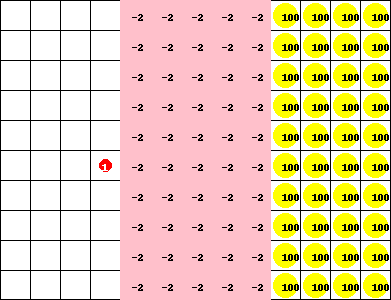}
    \caption{The initial state for the Wall of Fire task.}
    \label{fig:wall_of_fire}
\end{figure}

The Expert + Expert team plays the optimal strategy, which traverses the wall of fire in 5 steps then spends the remaining steps collecting coins. Conversely, the novice team is unable to ``see'' past the wall of fire, and remains in the neutral area to the left of the wall of fire. However, when the oblivious expert is paired with the novice, the two players take turns moving into and out of the fire for the entirety of the episode, leading to a large negative reward (Expert + Novice, \Cref{tab:wall_of_fire}). In contrast, the capability-aware expert is able to infer from a single interaction that the novice is unable to see past the wall of fire, and only collects the penalty once before cooperating with the novice to ``explore'' the neutral left side of the board (CA-Expert + Novice, \Cref{tab:wall_of_fire}).

\begin{table}[]
    \centering
    \begin{tabular}{c|c}
    \toprule
         Team Composition & Reward  \\
    \midrule
         Novice + Novice & 0 \\
         Expert + Expert & 1490 \\
    \midrule
         Expert + Novice & -20 \\
         CA-Expert + Novice & -2 \\
    \bottomrule
    \end{tabular}
    \caption{Performance on Wall of Fire task for different teams. Results are the median of 5 runs.}
    \label{tab:wall_of_fire}
\end{table}

In the \textbf{Narrow Tunnel} task, two players each control an avatar on a two-dimensional board. The action space consists of either not moving, or moving in a cardinal direction. The red coins are worth \mbox{+30} if collected by the red avatar, and the blue coins are worth \mbox{+1} if collected by by the blue avatar; the coins are worthless otherwise, and disappear regardless of which avatar collects them. We run each episode for 20 turns. We use a novice of depth 10, and an expert of depth 30. \Cref{fig:narrow_tunnel} displays the initial state of the players and the board. \Cref{tab:narrow_tunnel} reports the performance of various team compositions on the task.

\begin{figure}
    \centering
    \includegraphics[width=.5\columnwidth]{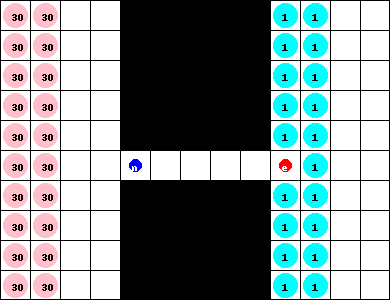}
    \caption{The initial state for the Narrow Tunnel task. In this case, the novice controls the blue avatar and the expert controls the red avatar.}
    \label{fig:narrow_tunnel}
\end{figure}

Due to the uneven rewards, the optimal strategy is for the blue avatar to yield the narrow tunnel to the red avatar, allowing it to collect the more valuable red coins (Expert + Expert, \Cref{tab:narrow_tunnel}). Conversely, the red rewards are out of reach for a novice player, so the novice team instead sends the blue avatar through the narrow tunnel (Novice + Novice, \Cref{tab:narrow_tunnel}). However, when an oblivious expert controls the red avatar, and the novice controls the blue avatar, the two players meet in a deadlock in the center of the tunnel, neither willing to yield to the other (Novice + Expert, \Cref{tab:narrow_tunnel}). An expert running the capability-aware algorithm takes only 1 turn of deadlock to infer that the novice's depth is insufficient to use the optimal strategy, and hence yields the tunnel to the blue player (Novice + CA-Expert, \Cref{tab:narrow_tunnel}).

\begin{table}[]
    \centering
    \begin{tabular}{c|c|c}
    \toprule
         Blue Avatar & Red Avatar & Reward  \\
    \midrule
         Novice & Novice & 4 \\
         Expert & Expert & 90 \\
    \midrule
         Novice & Expert & 0 \\
         Novice & CA-Expert & 4 \\
    \bottomrule
    \end{tabular}
    \caption{Performance on Narrow Tunnel task for different teams. Results are the median of 5 runs.}
    \label{tab:narrow_tunnel}
\end{table}

\subsection{Cooperative Checkers}

We next report results for a cooperative version of checkers, where players on the same team take turns moving pieces. Cooperative checkers presents a challenging setting to study ad hoc teamwork due to the complexity of the underlying game. For instance, \citet{schaeffer2007game} reports that an average game of checkers lasts for around 50 turns, with an average branching factor of around 6; the game-tree complexity is thus around $10^{40}$ (compared to $10^{83}$ for chess).

Each multi-agent team consist of a single expert paired with a single novice, with capabilities varying over all distinct pairs from the set $\{2, 4, 6, 8\}$. Players alternate selecting actions when it is the team's turn to move. The game ends when either a team has no moves or no pieces left, in which case the opponent wins. We also terminate a game after 120 total moves, or after 40 moves without any rewards, and declare the team with the highest cumulative rewards as the winner. Players plan without knowledge of the termination conditions to prevent stalling. The appendix contains further details about the set up and additional experimental results.

\subsubsection{Capability-Aware Versus Oblivious Cooperation}

For the first set of experiments, we play two teams against each other: capability-aware (CA) and oblivious (OBL). Both teams consist of an expert and a novice, with the experts and novices on either team having the same depths. The expert on team CA is capability-aware, while the expert on team OBL is oblivious; both novices are oblivious. We played 20 games with different random seeds for each assignment of colors to teams and player order within each team, for a total of 160 games per combination of capabilities. We report the aggregate \textbf{score}, defined as (\#wins - \#losses)/\#games, with a higher score indicating stronger performance.

\begin{table}[]
    \centering
    \begin{tabular}{c|c|c}
    \toprule
        $\Delta$ & \#Runs & Score (\%)  \\
    \midrule
         2 & 480 & 1.0 \\
         4 & 320 & 6.4 \\
         6 & 160 & 4.4 \\
    \midrule
         total & 960 & 3.3 \\
    \bottomrule
    \end{tabular}
    \caption{Performance of a CA team against a OBL team, scored from the perspective of the CA team (a positive score means CA is stronger than OBL). $\Delta$ denotes the difference between the expert and novice capabilities.}
    \label{tab:checkers_vs}
\end{table}

\Cref{tab:checkers_vs} displays the results. Team CA has a consistent advantage over team OBL, which tends to be larger when the difference in capabilities $\Delta$ is larger. This is consistent with the expectation that the oblivious expert mispredicts the novice's policy more frequently as $\Delta$ increases.

\subsubsection{Multi-Agent Capability-Aware Cooperation}
\label{sec:results:ma_vs_sa}

We also compared the performance of a team running the multi-agent capability inference algorithm (MA) against a team consisting of two adaptive players running a single-agent inference algorithm (SA). In the former case, each player maintains a capability type structure to model the teammate's beliefs and ensure consistent belief updates. In the latter case, each player tries to infer the capabilities of its teammate, but models the teammate as playing obliviously. While this simplifies the update rule by ignoring the non-stationarity of the teammate's behavior, our main hypothesis is that such an approximation leads to worse performance before the beliefs have converged the true capabilities.

Teams play against a single player opponent whose capability is equal to either the novice or the expert on the team. We ran 50 games with different random seeds for each assignment of colors to teams and player order within the team, for a total of 200 games per combination of capabilities and opponent. In addition to the score, we also report the \textbf{deviation} $d$ of the final posterior $p$ with respect to the true depth $d^*$ of the teammate, defined as $\sqrt{\sum_d p(d)(d - d^*)^2}$. In the case of the novice, we set $d^*$ to the best effort approximation (i.e., the novice's own depth). The ``true'' posterior placing 100\% density on $d^*$ achieves an optimal deviation of 0.

\begin{table}[]
    \centering
    \begin{tabular}{c|c|c|c|c|c}
    \toprule
        Opponent & Team & \#Runs & Score (\%) & $d_{exp}$ & $d_{nov}$\\
    \midrule
         \multirow{2}{*}{novice} & MA & 600 & 27.0 & 1.1 & 0.8\\
         & SA & 600 & 21.0 & 1.2 & 0.9 \\
    \midrule
         \multirow{2}{*}{expert} & MA & 600 & -19.7 & 1.2 & 0.8 \\
          & SA & 600 & -21.8 & 1.3 & 0.9 \\
    \midrule
    \midrule
         \multirow{2}{*}{total} & MA & 1200 & 3.7 & 1.1 & 0.8 \\
          & SA & 1200 & -0.4 & 1.2 & 0.9 \\
    \bottomrule
    \end{tabular}
    \caption{Performance of multi-agent (MA) and single-agent (SA) adaptive teams, against an opponent of either novice or expert depth. Scores are from the perspective of the team (a positive score means that the team is stronger than the single opponent). $d_{nov}$ and $d_{exp}$ denote the deviations of the final posteriors of the novice and expert, respectively.}
    \label{tab:checkers_ma}
\end{table}

\Cref{tab:checkers_ma} displays the results. The MA team running the correct multi-agent update consistently outperforms the SA team running the single-agent update, which is particularly noticeable when playing against the novice. We also note that the final posteriors are fairly accurate, even when using the single-agent update. Hence, we attribute the difference in performance to diverging posteriors earlier in the game: the posteriors may grow accurate even as the game is already lost. These results indicate that our multi-agent framework based on capability type structures yields fast convergence of the posteriors to the true values, which is critical to achieving good performance compared the naive single-agent update.
\section{Conclusion}

This work studies ad hoc teamwork amongst agents of different computational capabilities. Our main contribution is a framework based on capability type structures, which enables multiple agents of different computational capabilities to adapt their behavior without explicit coordination. Moreover, by using tempered beliefs our techniques can also be applied when the agents' observations and value estimates are subject to noise. Our experiments indicate that our capability-aware algorithm leads to improved performance for teams of heterogeneous capabilities in several toy settings, as well as a more complex checkers environment.

\bibliography{main}

\onecolumn
\appendix
\section{Deferred Proofs}

\begin{proof}[Proof of \Cref{thm:consistency}]
Let $i$ be the current actor, and let $j$ and $k$ be any two players with capability types $c_j, c_k$, respectively.

Fix a type $c \le c_j$ and $c \le c_k$. First, we show that players $j$ and $k$ compute the same value for $v^*(c)$, which is the optimal value at capability $c$. Denote by $v^j(c)$ and $v^k(c)$ the optimal values at capability $c$ computed by players $j$ and $k$, respectively. By assumption, $B_j(t)$ is a capability type structure for all $j \in \mathcal{N}$, so (by a slight abuse of notation) we have $B^j_{[i=c]} \downarrow_c = B^k_{[i=c]} \downarrow_c$. As $Q^c$ and $s$ are public knowledge,
\begin{align}
v^j(c) &= \max_{a' \in \mathcal{A}} Q^c(s, B^j_{[j=c]} \downarrow_c, a') \\
&= \max_{a' \in \mathcal{A}} Q^c(s, B^k_{[j=c]} \downarrow_c, a') \\
&= v^k(c).
\end{align}

Similar, let $A^j(c)$ and $A^k(c)$ be the optimal actions $A^*(c)$ as computed by players $j$ and $k$, respectively. Then
\begin{align}
A^j(c) &= \{a \in \mathcal{A} | Q^c(s, B^j_{[j=c]}\downarrow_c, a) = v^j(c)\} \\
&= \{a \in \mathcal{A} | Q^c(s, B^k_{[j=c]}\downarrow_c, a) = v^k(c)\} \\
&= A^k(c).
\end{align}

Again by the fact that $B_j$ forms a capability type structure, $B^j_i[c] = B^k_i[c]$, and
\begin{align}
B^j_i[c] * \mathbbm{1}_{a \in A^j(c)} / |A^j(c)| = B^k_i[c] * \mathbbm{1}_{a \in A^k(c)} / |A^k(c)|.
\end{align}

Since this is true for all $i,j \in \mathcal{N}$ and $c \in \mathcal{N}$ such that $c \le c_i$ and $c \le c_j$, we conclude that $B_j(t{+}1)$ is a capability type structure.

\end{proof}

\begin{proof}[Proof of \Cref{thm:correctness}]
This result follows directly from the fact that the belief update term $\mathbbm{1}_{a \in A^*(c')} / |A^*(c')|$ is exactly the likelihood $P(a|c')$. More explicitly, let $a(t_1), a(t_2), \ldots, a(t_T)$ be the sequence of actions played by player $j$ thus far. By Bayes' theorem,
\begin{align}
    B^j_i[c'] &= \frac{\mathbbm{1}_{a(t_1) \in A_1^*(c')}}{|A_1^*(c')|} \frac{\mathbbm{1}_{a(t_2) \in A_2^*(c')}}{|A_2^*(c')|} \cdots \frac{\mathbbm{1}_{a(t_T) \in A_T^*(c')}}{|A_T^*(c')|} \\
    &= P(a(t_1)|c_j = c') * P(a(t_2)|c_j = c') * \cdots * P(a(t_T)|c_j = c') \\
    &\propto P(c_j = c'|a(t_1), a(t_2), \ldots, a(t_T)).
\end{align}
where $A^*_n$ is the set of optimal actions at capability $c$ for player $i$ when it played action $a(t_n)$, and we have hidden the dependence on the history of state observations.

Next, we consider the denominator of $P^c(c')$. In particular, if $\sum_{c'' \le c} B_j[c''] = 0$, then $P^c$ is undefined, but this also means that $ P(c_j = c''|a(t_1), a(t_2), \ldots, a(t_T)) = 0$ for all $c'' \le c$, i.e., $c_j > c$ as claimed. On the other hand, if $\sum_{c'' \le c} B_j[c''] > 0$, then this computes the normalization factor for $P(\cdot | c_j \le c)$; hence $P^c(c') = P(c_j = c' | c_j \le c, a(t_1), a(t_2), \ldots, a(t_T))$, which completes the proof.
\end{proof}

Before we prove the approximate consistency results, we first state and prove a useful lemma:
\begin{lemma}
Let $i$ be the acting player, denote their action by $a$, and let $B^j$ and $B^k$ the denote the type $c$ beliefs of any two players $j$ and $k$ with type at least $c$. Denote by $\ell_a^j(c)$ and $\ell_a^k(c)$ the losses as computed by $j$ and $k$. If $|B^j_m[c'] - B^k_m[c']| / T(t) \le \epsilon / 2N$ for all players $m$ and all capabilities $c' \le c$, then $|\ell_a^j(c) - \ell_a^k(c)| \le 3\epsilon$.
\end{lemma}

\begin{proof}
First we show that $P^j_m$ and $P^k_m$ are close in total variation distance for arbitrary $m \in \mathcal{N}$. Fix $c' \le c$ and assume without loss of generality that $P^k_m(c') > P^j_m(c')$. Then
\begin{align}
P^k_m(c') - P^j_m(c') &= \frac{\exp(-B^j_m[c'] / T(t)}{\sum_{c'} \exp(-B^j_m[c']/ T(t))} - \frac{\exp(-B^k_m[c']/ T(t))}{\sum_{c'} \exp(-B^k_m[c']/ T(t))} \\
&\le \frac{\exp(-B^j_m[c']/ T(t))}{\sum_{c'} \exp(-B^j_m[c']/ T(t))} - \frac{\exp((-B^j_m[c'] - \epsilon/2T(t)N)/ T(t))}{\sum_{c''} \exp((-B^j_m[c''] + \epsilon/2T(t)N)/ T(t))} \\
&= (1 - e^{-\epsilon/N})\frac{\exp(-B^j_m[c']/ T(t))}{\sum_{c'} \exp(-B^j_m[c']/ T(t))} \\
&= (1 - e^{-\epsilon/N}) P^k_m(c') \\
&\approx  \epsilon P^k_m(c')/N 
\end{align}

Since this holds for all $c' \le c$, it follows that the total variation distance between $P_m^j$ and $P_m^k$ is at most
\begin{align}
d_{TV}(P^j_m, P^k_m) &= \sum_{c' \le c} |P^j_m(c') - P^k_m(c')|/2 \\
&\le \sum_{c' \le c} \epsilon P^k_m(c') /2N \\
&= \epsilon/2N
\end{align}

Denote by $\phi^c_{B^j}(C)$ and $\phi^c_{B^k}(C)$ the generalized likelihoods as computed by players $j$ and $k$, respectively. As each likelihood is a product of $N$ distributions $\prod_{m \in \mathcal{N}} P^j_m$ and $\prod_{m \in \mathcal{N}} P^k_m$, respectively, each bounded by in total variation distance by $\epsilon/2N$, we have that
\begin{align}
    d_{TV}(\phi^c_{B^j}(C), \phi^c_{B^k}(C)) \le \epsilon/2
\end{align}

Let $V^c_{\phi,j}(s, B)$ and $V^c_{\phi,k}(s, B)$ be the resulting $\phi$ value functions for players $j$ and $k$, respectively. Then
\begin{align}
    |V^c_{\phi,j}(s, B) - V^c_{\phi,k}(s, B)| &= |\mathbb{E}_{C \sim \phi^c_{B_j}}\big[ V_j^C(s_j) \big] - \mathbb{E}_{C \sim \phi^c_{B_k}}\big[ V_k^C(s_k) \big]| \\
    &\le \epsilon/2 + |\mathbb{E}_{C \sim \phi^c_{B_k}}\big[ V_j^C(s_j) \big] - \mathbb{E}_{C \sim \phi^c_{B_k}}\big[ V_k^C(s_k) \big]| \\
    &= \epsilon/2 + \mathbb{E}_{C \sim \phi^c_{B_k}}\big[ |V_j^C(s_j) -  V_k^C(s_k) |\big] \\
    &\le 3\epsilon/2
\end{align}

It follows then that the $\phi$-Q functions also differ by at most $3\epsilon/2$, and hence $|\ell_a^j(c) - \ell_a^k(c)| \le 3\epsilon$.
    
\end{proof}

\begin{proof}[Proof of \Cref{thm:consistency_approx_adv}]
\label{appendix:proof:consistency_approx_adv}
Let $B^j_i$ and $B^k_i$ be the type $c'$ beliefs of any two players $j$ and $k$. We will prove by induction that $|B^j_i[c] - B^k_i[c]| / t \le 3\epsilon$ for all $c \le c'$, where $t$ is the number of updates and $\epsilon$ is the error bound between private value functions. 

As the beliefs are initialized to the same values, the base case $t = 0$ is trivially true. For the inductive step, by assumption, $|B^j_i[c] - B^k_i[c]| / t \le 3\epsilon$, i.e., $|B^j_i[c] - B^k_i[c]| / T(t) \le \epsilon/2N$. Directly applying the lemma, we see that their estimates of the losses differ by at most $3\epsilon$. This completes the induction, as the beliefs accumulate the losses.

\end{proof}

\begin{proof}[Proof of \Cref{thm:consistency_approx_stoch}]

Let $B^j_i$ and $B^k_i$ be the type $c'$ beliefs of any two players $j$ and $k$.
Observe that $B^j_i[c]$ and $B^k_i[c]$ are random variables such that $\mathbb{E}[B^j_i[c]] = \mathbb{E}[B^k_i[c]]$ for all $j$ and $k$ due the assumption of independence of noise. Denote this expectation by $B_i[c]$. We will show that, with probability at least $1-\delta$, 
\begin{align}
\label{app:eq:stochastic_bound}
|B^j_i[c] - B_i[c]| / \sqrt{d}t^{2/3} \le \epsilon / 4N
\end{align}
for all $i, j \in \mathcal{N}$, $c \in \mathcal{C}$ and time steps $t$. It follows then that, with probability at least $1-\delta$, $|B^j_i[c] - B^k_i[c]| / \sqrt{d}t^{2/3} \le \epsilon / 2N$ for all $i, j, k \in \mathcal{N}$, $c \in \mathcal{C}$ and time steps $t$, which completes the proof.

We will proceed by inducting on $t$. As all beliefs are initialized to zero, the base case $t = 0$ holds. Hence assume that  $|B^j_i[c] - B^k_i[c]| / \sqrt{d}t^{2/3} \le \epsilon / 2N$, or equivalently, $|B^j_i[c] - B^k_i[c]| / T(t) \le \epsilon / 2N$. By the lemma, the losses differ from the true loss by at most $3\epsilon$. By Hoeffding's inequality
\begin{align}
P(|B^j_i[c] - B_i[c]| \ge \sqrt{d}\epsilon t^{2/3}/4N) &\le 2 \exp\big(-\frac{2d\epsilon^2t^{4/3}}{(4N)^2t(3\epsilon)^2}\big) \\
&= 2\exp\big(-d't^{1/3}\big)
\end{align}
where $d' = d/72N^2$ and $t$ is the number of loss estimates in the sum $B^j_i[c]$. There are $N^2$ beliefs $B^j_i$ (one for each $i,j \in \mathcal{N}$), each of which is has support at most $c_{max} := \max_{i \in \mathcal{N}} |p(c_i)|$. Hence, in order to apply the union bound, we need
\begin{align}
\sum_{n=1}^\infty 2\exp(-d't^{1/3}) \le \delta / N^2c_{max}.
\end{align}

As the corresponding integral converges, we find that taking $d' = \ln(20N^2c_{max}/9\delta)$ suffices.
We conclude by the union bound that, with probability at least $1-\delta$, \Cref{app:eq:stochastic_bound} holds for all agents $i,j \in \mathcal{N}$, all $c \in \mathcal{C}$, and all timesteps $t$. 

\end{proof}
\section{Capability-Aware Monte-Carlo Tree Search}
\label{appendix:trees}

The Monte-Carlo Tree Search (MCTS) algorithm has seen a surge in popularity for games and planning in recent years since forming an integral part of AlphaGo \citep{silver2016mastering}.  This section provides a brief description of the MCTS algorithm. We next present a modification of MCTS that accommodates agents that are bounded in how deep into the game tree they are able to explore. We then propose a capability-aware MCTS algorithm that adjusts for the capabilities of teammates. Finally, we describe how we perform inference in the capability-aware MCTS algorithm.

\subsection{Monte-Carlo Tree Search}

In this section, we provide a brief introduction to the MCTS algorithm. For a more comprehensive survey, we refer the reader to \citep{browne2012survey}. 

\Cref{alg:mct} provides an overview of the main steps in MCTS. The search is initialized with the tree consisting of a single root node $\mathcal{T} = s(0)$ consisting of the current state. At each step, we $\textbf{Select}$ a path to a leaf node $s$, $\textbf{Expand}$ the leaf node by adding all legal transitions as children, $\textbf{Simulate}$ an initial value estimate of $s$, then $\textbf{Backpropogate}$ the estimated value of $s$ to the root, using the recursive formula to update the value of its predecessors. In particular, $\textbf{Simulate}$ performs $m$ simulations of random actions down to a maximum depth of $d$ (or until the game terminates), and takes the average cumulative reward to be the initial value estimate. $\textbf{Select}$ is subject to the usual explore-exploit tradeoff (exploration prefers nodes which have fewer visits; exploitation prefers nodes with higher estimated values); the popular Upper Confidence Trees (UCT) variant \citep{kocsis2006bandit} treats each choice as a multi-armed bandit problem, using the celebrated UCB1 method \citep{auer2002finite} to achieve good theoretical guarantees (and excellent performance in practice). 

\begin{algorithm}
\caption{Monte-Carlo Tree Search (MCTS)}
\label{alg:mct}
\begin{algorithmic}[1]

\Procedure{MCTS}{$\mathcal{T}, n, d, m$}
\For {$i = 1, ..., n$}
    \State $s \gets \textbf{Select}(\mathcal{T})$
    \State $\textbf{Expand}(s)$
    \State $\textbf{Simulate}(s, d, m)$
    \State $\textbf{Backpropogate}(\mathcal{T}, s)$
\EndFor
\EndProcedure
\end{algorithmic}
\end{algorithm}

\subsection{Depth-Bounded MCTS}

To adjust the core MCTS algorithm for agents with bounded search depth $d$, we add a restriction that \textbf{Select} terminates the search at either a leaf, or when the path as reached the maximum depth $d$. As the same parameter $d$ also applies to the simulation depth, an agent of depth $d$ is able to explicit store a tree of depth $d$, and can access the next $d$ levels via Monte Carlo simulations. The updated algorithm is provided in \Cref{alg:mct_depth}; note that the only difference compared to \Cref{alg:mct} is the parameter $d$ in \textbf{Select}.

\begin{algorithm}
\caption{Depth-Bounded MCTS}
\label{alg:mct_depth}
\begin{algorithmic}[1]

\Procedure{BoundedMCTS}{$\mathcal{T}, n, d, m$}
\For {$i = 1, ..., n$}
    \State $s \gets \textbf{Select}(\mathcal{T}, d)$
    \State $\textbf{Expand}(s)$
    \State $\textbf{Simulate}(s, d, m)$
    \State $\textbf{Backpropogate}(\mathcal{T}, s)$
\EndFor
\EndProcedure
\end{algorithmic}
\end{algorithm}

Finally, we assume that agents search through the depths of the tree progressively. This additional structure simplifies the inference and empirically does not significantly affect the strength of the agent in our setting. \Cref{alg:mct_agent_obl} presents the algorithm an oblivious agent with capability type $d$ uses to perform the search. Note that the agent also spends an increasing amount of time searching at each level.

\begin{algorithm}
\caption{Depth-Bounded MCTS Agent (Oblivious)}
\label{alg:mct_agent_obl}
\begin{algorithmic}[1]

\Procedure{ObliviousSearch}{$\mathcal{T}, n, d, m$}
\For {$i = 1, ..., d$}
    \State \textsc{BoundedMCTS}($\mathcal{T}, n*i, i, m$)
\EndFor
\EndProcedure
\end{algorithmic}
\end{algorithm}

\subsection{Capability-Aware MCTS}

In general, we can separate the nodes of game tree into two categories: those that represent actions by the player's own team, and those that are played by the opponent. The standard technique for handling the opponent's nodes is perform a min-max search, where the opponent is presumed to play the action with the worst value (from the player's perspective). Conversely, for nodes played by the player's team, the naive strategy (adopted by the oblivious player) is to assume that the teammate will select the best action. However if the node is actually to be played by a teammate with different capabilities, then they may disagree with your assessment of which action to play, causing incorrect values to be backpropogated.

To address this, we propose to modify the MCTS algorithm to maintain one value (and visit count) per capability type, rather than the single statistics used in the standard algorithm. In particular, each time the player progresses to the next level of search depth, it also moves to storing the values for the type corresponding to that search depth. The values are initialized (for depths > 1) to those of the previous depth.

To leverage the typed values, at the beginning of each iteration in the \textsc{BoundedMCTS}, we sample, according to the beliefs, a set of depths for the teammates. During selection, we continue to use the best (i.e., deepest) estimates of the value for our own nodes and the opponent's nodes, but for a teammate's node we use the value corresponding to the sampled type. Furthermore, we also recursively run another instance of \textsc{BoundedMCTS} with the teammate's hypothesized depth to simulate exploring as the teammate. To limit the amount of overhead due to the recursive simulation, we only allow the initial call to spawn additional simulations. Hence, the modified MCTS algorithm introduces overhead linear in the original agent's depth.

\Cref{alg:mct_ca} presents the capability-aware versions of the core MCTS algorithm, which access an additional input $B$ representing the beliefs of the agent. The new subroutine \textbf{SampleDepths} samples a set of depths for the agent's teammates based on the beliefs, \textbf{Select} additionally takes the sampled depths $D$ as input, and \textbf{Backpropogate} uses the current maximum depth $d$ to store the values for the appropriate type.

\begin{algorithm}
\caption{Capability-Aware MCTS}
\label{alg:mct_ca}
\begin{algorithmic}[1]

\Procedure{CA-MCTS}{$\mathcal{T}, n, d, m, B$}
\For {$i = 1, ..., n$}
    \State $D \gets \textbf{SampleDepths}(B\downarrow_d )$
    \State $s \gets \textbf{Select}(\mathcal{T}, d, D)$
    \State $\textbf{Expand}(s)$
    \State $\textbf{Simulate}(s, d, m)$
    \State $\textbf{Backpropogate}(\mathcal{T}, s, d)$
\EndFor
\EndProcedure
\end{algorithmic}
\end{algorithm}

Finally, we describe how we implement inference. Each time before a teammate takes an action, the player first runs a fresh instance of CA-MCTS from the current state. Because of the progressive nature of the search, pausing the search after a depth $d$ provides an estimate of the behavior of an agent of depth $d$. Hence, the player can simply read off the $Q$ values for each capability type since they are already stored in each node. Then after observing an action $a$ by the teammate, the player can compute the losses $\ell_a(d)$ to updates its beliefs. Note that the standard UCT action selection picks the most visited action, so in practice we substitute visit counts for the average value in the $Q$ function.

\section{Additional Experimental Results and Details}

In this section, we provide additional experiments result and details for the checkers environment.

\subsection{Hyperparameters}

For MCTS, all agents use $n=200$ for their search parameter (so that a depth $d$ agent searches for $100d(d+1)$ iterations in total), and $m=5$ for the number of rollouts. We use a discount rate of $0.9$ when planning, but use undiscounted cumulative rewards when determining the winner. For the belief updates, we set the temperature $T(t) = .1$, and additionally clip the losses to a maximum of $0.5$. All beliefs for the capability-aware agents are initialized to a uniform distribution (unless otherwise noted).

\subsection{Ablation Studies}

We performed an additional evaluation of our capability-aware MCTS algorithm with the following baselines:
\begin{enumerate}
    \item Capability-Aware (\textbf{CA}): Runs CA-MCTS with beliefs initialized to the uniform distribution, and performs belief updates.
    \item Oracle (\textbf{ORA}): Runs CA-MCTS with beliefs initialized to the true beliefs, and does not perform any belief updates.
    \item Oblivious (\textbf{OBL}): Runs the standard MCTS.
    \item No-Update (\textbf{NU}): Runs CA-MCTS with beliefs initialized to the uniform distribution, and does not perform any belief updates.
    \item Minimizer (\textbf{MIN}): Runs the standard MCTS, but models teammates as opponents (i.e., value minimizers).
    
\end{enumerate}

Each team consists of two players: an expert using one of the 5 strategies described above, and an oblivious novice. The setup is otherwise identical to the one in \Cref{sec:results:ma_vs_sa}: the depths of the novice and expert are drawn (without replacement) from $\{2, 4, 6, 8\}$, and they play against a single player whose depth is equal to either the novice or the expert.

\begin{table}[]
    \centering
    \begin{tabular}{c|c|c|c}
    \toprule
        Opponent & Team & \#Runs & Score (\%) \\
    \midrule
         \multirow{5}{*}{novice} & CA & 600 & \textbf{26.0} \\
         & ORA & 600 & 25.3 \\
         & OBL & 600 & 12.8 \\
         & NU & 600 & 19.0 \\
         & MIN & 600 & 15.2 \\
    \midrule
         \multirow{5}{*}{expert} & CA & 600 & -21.2 \\
         & ORA & 600 & \textbf{-14.5} \\
         & OBL & 600 & -17.7 \\
         & NU & 600 & -18.3 \\
         & MIN & 600 & -17.8 \\
    \midrule
    \midrule
         \multirow{5}{*}{total} & CA & 1200 & 2.4 \\
         & ORA & 1200 & \textbf{5.4} \\
         & OBL & 1200 & -2.4 \\
         & NU & 1200 & 0.3 \\
         & MIN & 1200 & -1.3 \\
    \bottomrule
    \end{tabular}
    \caption{Performance of oracle (\textbf{ORA}), capability aware (\textbf{CA}), oblivious (\textbf{OBL}), no update (\textbf{NU}), and min team (\textbf{MIN}) strategy experts when paired with an oblivious novice, playing against an opponent of either novice or expert depth. Scores are from the perspective of the team (a positive score means that the team is stronger than the single opponent).}
    \label{tab:checkers_baseline}
\end{table}

\Cref{tab:checkers_baseline} reports the results. 
In aggregate, the oracle performs the best, with a total score of 5.4, which is to be expected. This result confirms that our capability-aware MCTS algorithm is able to correctly adjust its behavior for a teammate's depth.

Our capability-aware agent with inference performs second best with a total score of 2.4, performing as well as the oracle against a novice agent. However, the team underperforms against an expert, which we attribute to mispredictions prior to convergence of beliefs. Note that each game is a challenging ``one-shot'' setup, where the beliefs do not carry over between games. In a repeated game setting, we would expect the beliefs to converge, and the difference between the oracle and capability-aware teams to disappear.

The worst performing team is the oblivious team, which confirms our motivation that being capability-aware is important in cooperative settings with heterogeneous capabilities. The team with the expert that models the novice as an opponent performs poorly as well, due to playing too conservatively. Finally, we note that simply using a uniform prior and playing with the CA-MCTS algorithm yields marginally better results than being either oblivious or completely minimizing. This strategy corresponds to an agent who adapts its behavior for the \textit{possibility} that its teammate will have different capabilities, but does not attempt to infer the true depth.

\end{document}